\documentclass[12pt,a4paper,useAMS]{article}
\usepackage{graphicx}
\usepackage{amsmath,amssymb,amsthm,amsfonts}
\usepackage{bbold}
\usepackage{color, colortbl}
\usepackage{enumitem}
\usepackage{subfigure}
\usepackage{multirow,natbib}
\newtheorem{cond}{Conditions}
\newtheorem{theorem}{Theorem}

\newtheorem{prop}{Proposition}
\newtheorem{rem}{Remark}
\newtheorem{cor}{Corollary}
\newtheorem*{keywords}{Keywords}
\newtheorem{definition}{Definition}
\DeclareMathOperator{\diag}{diag}

\newcommand{\obs}{y}
\newcommand{\missing}{\varphi}
\newcommand{\meanrandom}{\beta}
\newcommand{\varrandom}{\Gamma}
\newcommand{\nbInd}{N}
\newcommand{\nbMiss}{p}

\newcommand{\nbMeasPerInd}{J}

\newcommand{\covmatrixerr}{\Sigma}

\graphicspath{{images/}}

\usepackage[figuresright]{rotating}

\begin{document}

\title{Likelihood ratio test for variance components in nonlinear mixed effects models.}

\author{C. Baey \\
		Universit\'e de Lille 1, Laboratoire Paul Painlev\'e\\
		Cit\'e Scientifique, 59655 Villeneuve d'Asq Cedex, France
		\and P.-H. Courn\`ede \\
	   CentraleSup\'elec, Laboratoire MICS \\
	   8, Rue Joliot-Curie, 92190 Gif sur Yvette, France
	   \and 
	   E. Kuhn\\
	   INRA, MaIAGE\\
	    Domaine de Vilvert, 78352 Jouy-en-Josas, France
	   }
%


\maketitle

\begin{abstract}
Mixed effects models are widely used to describe heterogeneity in a population. A crucial issue when adjusting such a model to data consists in
identifying fixed and random effects. From a statistical point of view, it remains to test the nullity of the variances of a given subset of random effects.   Some authors have proposed to use the likelihood ratio test and have established its asymptotic distribution in some particular cases. Nevertheless, to the best of our knowledge, no general variance components testing procedure has  been fully investigated yet. In this paper, we study the likelihood ratio test properties to test that the variances of a general subset of the random effects are equal to zero in both linear and nonlinear mixed effects model, extending the existing results. 
 We prove that the asymptotic distribution of the test is a chi-bar-square distribution, that is to say a mixture of chi-square distributions, and we identify the corresponding weights. We highlight in particular that the limiting distribution depends on the presence of correlations between the random effects but not on the linear or nonlinear structure of the mixed effects model. We illustrate the finite sample size properties of the test procedure through simulation studies and apply the test procedure to two real datasets of dental growth and of coucal growth.
\end{abstract}

\begin{keywords}
Chi-bar-square distribution, inference under constraints, hypothesis testing, likelihood ratio test, nonlinear mixed effects models, variance components
\end{keywords}

\section{Introduction}
Mixed effects models have been extensively used in population models in order to account for heterogeneity in populations and to describe the intra and inter-individual variability  (see \cite{Pin00,Dav03,Lavielle2014}). 
In mixed effects models, parameters are of two types: on one side, fixed effects that are common to all the individuals of the population; on the other side, random effects that vary from one individual to the other. The last ones are also called individual parameters.

From a modelling point of view, a key question when adjusting a population model to a dataset is to identify the fixed and random effects of the model. From a statistical point of view, it can be rephrased as a test on the nullity of the variances of a given subset of all the random effects. Several authors have been interested in likelihood ratio tests (LRT) in this context. However the standard theoretical results on the asymptotic distribution of the likelihood ratio test have been established under the assumption that the parameter space 
is open.  In the context of linear mixed models,  this assumption is fulfilled when considering some random effects variances being equal to zero, as long as the sum of the random effects variances and of  the residual variance remains positive. Nevertheless, this assumption is generally not fulfilled  in the context of nonlinear mixed effects models. Indeed, testing for null variances of random effects remains to test parameter values lying on the boundary of the parameter space which is no more open. In this setting, the standard theoretical results on the asymptotic distribution of the likelihood ratio test can not be applied  (see \cite{Sil11}) and one should resort to constrained statistical testing. We recall here some of the main existing results.
Let us denote by $\Theta$ the  parameter space, by $\Theta_0$ the subset of $\Theta$ corresponding to the null hypothesis and by $\Theta_1$ the subset corresponding to the alternative hypothesis, with $\Theta_0 \subset \Theta_1 \subset \Theta$.
\cite{Cher54}, assuming that $\Theta$ is open, treated the case where the true value of the parameter lies on the boundary of $\Theta_0$ and $\Theta_1$, which is assumed to be a proper set of $\Theta$, i.e. strictly contained in $\Theta$.
He provided a representation of the asymptotic distribution of the LRT and proved that it is asymptotically equivalent to testing the mean of a multivariate Gaussian distribution based on one single observation.
 A few years later, \cite{Chant74} generalized these results by considering the case where the true value lies in a subset of $\Theta$ which may not be a proper subset. \cite{Sha85} studied the asymptotic distribution of a larger class of tests when the true value is on the boundary of $\Theta_0$ and an interior point of $\Theta$. He established that the asymptotic distribution is a mixture of chi-square distributions.
Simultaneously, \cite{Self87} obtained similar results in the case where the true value is on the boundary of $\Theta$.  They proved in particular that the limiting distribution of the LRT for testing that the variance of one single random effect is equal to 0 is a mixture $\frac{1}{2} \delta_0 + \frac{1}{2} \chi^2_1$, where $\delta_0$ is the Dirac distribution at 0 and $\chi^2_1$ is the chi-square distribution with $1$ degree of freedom. These results were then extended to the general case of $M$-estimation by \cite{Gey94}.%
Building upon those works, several authors have addressed the issue of variance components testing in the specific context of mixed effects models.
\cite{stramlee94, stramlee95} proposed a likelihood ratio test procedure for linear mixed effects models, and identified the limiting distribution of the LRT statistics in some cases. In particular, they suggested that the limiting distribution might in fact be influenced by the presence of correlations between random effects.
Some authors have also proposed finite sample test procedures for variance components testing in linear mixed models, either by deriving the finite sample distribution of test statistics, or using bootstrap and permutation tests.
For example, the finite sample size distribution of the likelihood and restricted likelihood ratio test statistics was studied by \cite{Crai04a} for linear mixed models with one single random effect, and \cite{Grev08} extended these results to linear mixed models with more than one random effect. Several years later, \cite{Qu2013} proposed a procedure based on the score test for testing several variance components in linear mixed models, and \cite{Wood13} studied the finite sample distribution of a test based on the restricted likelihood for testing that one variance is null in generalized linear mixed models. Also in the context of linear mixed models,  
\cite{Sin09} studied a bootstrap test based on the score test for testing several variance components in a generalized linear mixed model, while \cite{Fitz07}, \cite{Sam12} and \cite{Dri13} considered permutation tests for testing several variance components in the context of linear and generalized mixed effects models.

\cite{Mol07} proposed a review of the existing results for testing variance components in mixed effects models, and studied in particular the equivalence between the LRT, the Score test and the Wald test, based on results by \cite{Sil95} or \cite{stramlee94}. They also exhibited the common limiting distribution in some specific cases. However, to the best of our knowledge, there exist no results identifying the limiting distribution of the LRT for general tests on variance components in mixed effects models. 

In this paper, we study the LRT in general mixed effects models, to test that the variances of any subset of the random effects are equal to zero, and identify its asymptotic distribution as a mixture of chi-square distributions. In Section \ref{sec:NLMEmodel}, we present the framework of nonlinear mixed effects models. Section \ref{sec:test} is devoted to the description of the proposed test and its theoretical properties. Practical implementation guidelines are presented in Section \ref{sec:practical}. Experimental results illustrate the performances of the procedure through simulation studies and real datasets analysis in Section \ref{sec:numerical}. The paper ends with some discussion in Section \ref{sec:discussion}. The technical proofs are given in  Appendix.

\section{Nonlinear mixed effects model}
\label{sec:NLMEmodel}
\subsection{Description}\label{sec:defNLME}
We consider the following nonlinear mixed effects model \citep{Dav03,Lavielle2014}:
\begin{align}\label{eq:intravar}
y_{i} = g(\varphi_i,x_{i}) + \varepsilon_{i},
\end{align}
where $y_{i}$ denotes the vector of observations of the $i$-th individual of size $\nbMeasPerInd$, $1 \leq i \leq \nbInd$, $g$ a nonlinear function, $\missing_i$  the vector of random effects of individual $i$, $x_{i}$ a vector of covariates, and $\varepsilon_i$ the random error term.

The vectors of random effects $(\varphi_i)_{1 \leq i \leq \nbInd}$ are assumed independent and identically distributed as follows:
\begin{align}\label{eq:inter_var}
\varphi_i \sim \mathcal{N}_\nbMiss(\meanrandom,\varrandom), \ 1 \leq i \leq \nbInd,
\end{align}
where $\beta$ is a parameter vector in $\mathbb{R}^p$, and $\Gamma$ a covariance matrix of size $p \times p$.

The vectors $(\varepsilon_{i})_{\substack{1 \leq i \leq \nbInd}}$ are assumed independent and identically distributed as follows:
\begin{equation}
\varepsilon_{i} \sim \mathcal{N}(0,\covmatrixerr).
\end{equation}
The sequences $(\varepsilon_{i})$ and $(\missing_i)$ are assumed mutually independent.

Let us denote by $\theta=(\meanrandom,\varrandom, \Sigma)$ the vector of all the model parameters and by $q$ its dimension. Thus, the parameter space is defined as $\Theta = \mathbb{R}^p \times \mathbb{S}^p_+ \times \mathbb{S}^{\nbMeasPerInd}_+$, where $\mathbb{S}^p_+$ is the set of symmetric, positive semi-definite $p \times p$ matrices.

\subsection{Examples}\label{sec:examples}
\subsubsection{Linear mixed effects model}
A special but very common case is the one where the function $g$ is linear. 
 The model can be rewritten in the following usual form \citep{Pin00}:
\begin{align}\label{eq:linearmodel}
y_i & = X_i \beta + Z_i \varphi_i + \varepsilon_i,
\end{align}
where $y_i$ is the observation vector for indivual $i$, $X_i$ and $Z_i$ are matrices of known covariates, $\beta$ is the vector of fixed effects, $\varphi_i$ is the vector of centered random effects for individual $i$, with $\varphi_i \sim \mathcal{N}(\mathbf{0}, \Gamma)$, and $\varepsilon_i$ is a random error term, with $\varepsilon_i \sim \mathcal{N}(\mathbf{0}, \Sigma)$.

\subsubsection{Nonlinear growth curve model}
One famous example of a nonlinear mixed effects model is the logistic growth model, which was studied for example by \cite{Pin00} in their well known example of orange trees growth.

In this model, a logistic curve is used to model the growth of each individual in the population as a nonlinear function of three individual parameters. Denoting by $y_{ij}$ the variable measured for individual $i$ at age $x_j$ (e.g. the trunk circumference in the orange trees example), for each individual $i$ these three parameters are: the asymptotic value of $y_{ij}$, $\varphi_{i1}$, the age at which the individual reaches half its asymptotic value, $\varphi_{i2}$, and the growth scale $\varphi_{i3}$. More precisely, we have:
\begin{eqnarray}
\label{eq:orangegrowth}
  y_{ij}&=&   \frac{\varphi_{i1}}{1+\exp\left(-\frac{x_{j}-\varphi_{i2}}{\varphi_{i3}}\right)}+\varepsilon_{ij}
  \quad 1\leq i \leq \nbInd, \ 1 \leq j \leq J , \\
\end{eqnarray}
where  $\varphi_{i}\sim\mathcal {N}_3(\beta,\Gamma)$, $\varepsilon_{ij}\sim\mathcal {N}(0,\sigma^{2})$ and the $(\varepsilon_{ij})_{i,j}$ are independent.

\section{Variance components testing}\label{sec:test}

\subsection{Description of the testing procedure }\label{subsec:desctest}
Let $r \in \{1, \dots, \nbMiss \}$.
We consider general test hypotheses of the following form, to test the nullity of $r$ variances and of the corresponding covariances:
\begin{equation}\label{eq:hypGeneral}
H_0 : \theta \in \Theta_0 \quad \text{against} \quad H_1 : \theta \in \Theta,
\end{equation}
where $\Theta_0 \subset \Theta$. 

Up to permutations of rows and columns of the covariance matrix $\Gamma$, we can assume that we are testing the nullity of the  last $r$ variances. We write $\Gamma$ in blocks as follows:
\begin{equation*}
\Gamma = \left( \begin{array}{c|c}
\Gamma_1 &  \Gamma_{12}^t \\
\hline
\Gamma_{12} & \Gamma_2
\end{array}
\right),
\end{equation*}
with $\Gamma_1$ a $(p-r) \times (p-r)$ matrix, $\Gamma_2$ a $r \times r$ matrix, $\Gamma_{12}$ a $r \times (p-r)$ matrix and where $A^t$ denotes the transposition of matrix $A$, for any matrix $A$.

The spaces associated to the null and alternative hypotheses are then:
\begin{align}\label{eq:paramSpaces}
	\Theta_0 & = \{\theta \in  \mathbb{R}^q \mid \beta \in \mathbb{R}^p, \Gamma_1 \in \mathbb{S}_+^{p-r}, \Gamma_{12} = 0, \Gamma_2 = 0, \covmatrixerr \in \mathbb{S}^{\nbMeasPerInd}_+ \} \\
	\Theta & = \{\theta \in  \mathbb{R}^q \mid \beta \in \mathbb{R}^p, \Gamma \in \mathbb{S}_+^{p}, \covmatrixerr \in \mathbb{S}^{\nbMeasPerInd}_+ \}.
\end{align}

We emphasize that the parameter space $\Theta$ is not open, and that the tested parameter values are on the boundary of $\Theta$.

We recall below the likelihood ratio test procedure. Let us denote by $y_1^{\nbInd} $ the joint vector of a $\nbInd$-sample $(y_1,\dots,y_{\nbInd})$, and by $L (y_1^{\nbInd} ; \theta)$ the joint likeli\-hood.
We then define the likelihood ratio test statistics by:
\begin{equation}\label{eq:lrtStat}
	LRT_\nbInd  := -2 \ \log \left( \frac{\sup_{\theta \in \Theta_0}  L (y_1^{\nbInd} ;\theta)}{\sup_{\theta \in \Theta}L (y_1^{\nbInd} ; \theta)} \right).
\end{equation}

For a nominal level $0 < \alpha < 1$, the rejection region $R_{\alpha}$  is defined by
\begin{equation}\label{eq:regionrejet}
R_{\alpha} = \{ LRT_N \geq q_{\alpha} \},
\end{equation}
where $q_{\alpha}$ is the $(1-\alpha)$ quantile of the distribution of $LRT_\nbInd$ under the null hypothesis.

However in practice the finite sample distribution of $LRT_N$ is generally untractable in the case of nonlinear mixed effects models. Therefore we focus on its asymptotic distribution.

\subsection[Asymptotic properties]{Asymptotic properties of the likelihood ratio test}\label{sec:properties}
Let us denote by $\theta^*$ the true value of the parameters. We assume that the following conditions are satisfied:

\begin{cond}\label{cond:alternative}
~
\begin{enumerate}
    \item the value $ \theta^*$  is in $\Theta_0$, i.e. $\theta^*$ is of the form $\theta^* = (\beta^*, \Gamma^*, \Sigma^*)$, with $\Gamma^* =   \left( \begin{array}{c|c}
\Gamma_1^* & 0 \\
\hline
0 & 0
\end{array}
\right)$.
\item the matrices $\Gamma_{1}^*$ and $\Sigma^*$ are positive definite
\end{enumerate}
In particular, we assume that the variances that are not being tested are strictly positive.
\end{cond}

To establish the asymptotic distribution of the likelihood ratio test statistics under the null hypothesis, we need to ensure the consistency of the maximum likelihood estimate (MLE). Therefore, we assume that the following general conditions are fulfilled \citep{Sil11}:
\begin{cond}\label{cond:likelihood}
~
\begin{enumerate}
  \item the function $L$ is injective in $\theta$ (to ensure the identifiability of the model),
    \item the first three derivatives of the log-likelihood w.r.t. $\theta$ exist and are bounded by a function whose expectation exists,
    \item the Fisher information matrix is finite and positive definite.
\end{enumerate}
\end{cond}

\begin{rem}
Note that the consistency and asymptotic normality of the MLE models in the specific context of nonlinear mixed effects have been studied in \cite{Nie06,Nie07}. He exhibited specific assumptions that ensure these theoretical results. However, they might be difficult to verify in practice.
\end{rem}

Before stating the expression of the asymptotic distribution of the likelihood ratio test statistics, we recall the definition of the chi-bar-square distribution (for more details, see \cite{Sha85,Sil11}).
\begin{definition}\label{def:chibar}
Let $\mathcal{C}$ be a closed convex cone of $\mathbb{R}^q$, $V$ a positive definite matrix of size $q \times q$ and $Z \sim \mathcal{N}(0,V)$. The distribution of the random variable defined by 
\begin{equation}\label{eq:chibarsquareRV}
\bar{\chi}^2(V,\mathcal{C}) = Z^t V^{-1}Z-\min_{\theta \in \mathcal{C}}(Z-\theta)^t V^{-1}(Z-\theta)
\end{equation}
is called a \emph{chi-bar-square distribution}. It is equal to a mixture of chi-square distributions with different degrees of freedom as follows:
\begin{equation}\label{eq:chibarsquareFDR}
\forall t \geq 0 \ P( \bar{\chi}^2(V,\mathcal{C}) \leq t)= \sum_{i=0}^q w_i(q,V,\mathcal{C})P(\chi_i^2 \leq t),
\end{equation}
where the weights $(w_i(q,V,\mathcal{C}))_{0 \leq i \leq q}$ are non-negative numbers summing up to one, and where $\chi_i^2$ is a random variable following the chi-square distribution with $i$ degrees of freedom, with the convention that $\chi_0^2 \equiv 0$.
\end{definition}

We can now establish the asymptotic distribution of the likelihood ratio test statistics.

\begin{theorem}\label{th:lrt}
Assume that conditions (\ref{cond:alternative}) and (\ref{cond:likelihood})  are fulfilled. We consider the test defined in (\ref{eq:hypGeneral}). We denote by $I_*$ the Fisher information matrix evaluated at the true value $\theta^* \in \Theta_0$. Then:
\begin{gather}\label{eq:lrtdist}
	LRT_\nbInd \xrightarrow[\nbInd  \rightarrow \infty]{d} \bar{\chi}^2(I^{-1}_*,T(\Theta,\theta^*)\cap T(\Theta_0,\theta^*)^{\perp}),
\end{gather}
where $T(\Theta,\theta)$ is the tangent cone to $\Theta$ at $\theta$, and $S^{\perp}$ is the orthogonal complement of $S$, for any subset $S$ of $\mathbb{R}^q$.
\end{theorem}

The proof is adapted from the proof of Proposition 4.8.2. in \cite{Sil11}. We present below a sketch of the proof, while its detailed description is postponed to the Appendix.
We first prove the consistency of the maximum likelihood estimates on $\Theta_0$ and $\Theta$. Then, we define the tangent cones to $\Theta_0$ and $\Theta$ at the true value $\theta^*$, and substitute $\Theta_0$ and $\Theta$ for their tangent cones. The MLE are still consistent on the tangent cones. Under some regularity conditions, we derive quadratic expansions of the log-likelihood around $\theta^*$. Finally, we prove that the asymptotic distribution of the LRT statistics defined in \eqref{eq:lrtStat} is the same as the asymptotic distribution of the LRT statistics when testing $\theta \in T(\Theta_0,\theta^*)$ against $\theta \in T(\Theta,\theta^*)$, based on a single observation of $Z \sim \mathcal{N}(\theta^*,I^{-1}_*)$. In other words, we are reduced to a test on the mean of a multivariate normal distribution.
Finally, since $T(\Theta_0,\theta^*$) is a linear space which is included in $T(\Theta,\theta^*)$, a closed convex cone, it follows using Theorem 3.7.1 in \cite{Sil11}, that the asymptotic distribution is a chi-bar-square distribution.

\bigskip
In the context of our study, the cone $T(\Theta,\theta^*)\cap T(\Theta_0,\theta^*)^{\perp}$ always admits an analytical expression.
The following proposition details this expression for $p$-dimensional random effects with non-correlated components or with fully-correlated components. These common cases correspond to a parameter space $\Theta$ involving covariance matrices $\Gamma$ which are either diagonal (with dimension of the parameter space $q=2p+\frac{J(J+1)}{2}$) or full ($q=p+\frac{p(p+1)}{2}+\frac{J(J+1)}{2}$). 

\begin{prop}\label{prop:cone}
$\quad$
\begin{enumerate}
	\item Assume that $\Theta = \{\theta \in  \mathbb{R}^q \mid \beta \in \mathbb{R}^p, \Gamma \in \mathbb{S}_+^{p}, \Gamma \ \text{full}, \covmatrixerr \in \mathbb{S}^{\nbMeasPerInd}_+ \}$. Then $$T(\Theta,\theta^*)\cap T(\Theta_0,\theta^*)^{\perp} = \{0\}^p  \times \{0\}^\frac{(p-r)(p-r+1)}{2}  \times \mathbb{R}^{r(p-r)} \times \mathbb{S}_+^{r}  \times \{0\}^{\frac{J(J+1)}{2}}.$$
	\item Assume that $\Theta = \{\theta \in  \mathbb{R}^q \mid \beta \in \mathbb{R}^p, \Gamma \in \mathbb{S}_+^{p}, \Gamma \ \text{diagonal}, 
\covmatrixerr \in \mathbb{S}^{\nbMeasPerInd}_+  \}$. Then $$T(\Theta,\theta^*)\cap T(\Theta_0,\theta^*)^{\perp} = \{0\}^p  \times \{0\}^{p-r} \times \mathbb{R}_+^{r}  \times \{0\}^{\frac{J(J+1)}{2}}.$$
\end{enumerate}
\end{prop}

The proof of Proposition \ref{prop:cone} relies on technical elements from convex analysis (see \cite{Hir12}) and is postponed to the Appendix. Results for covariance matrices of general structures can be easily proven using similar tools. 
For example, one can consider a block-diagonal structure for $\Gamma$. In particular, when dealing with mechanistic models each effect has a physical interpretation, and we can identify sub-groups of correlated random effects. In this case, up to a permutation of rows and columns, the covariance matrix can be written as $\Gamma = \diag(\Gamma_1,\dots,\Gamma_K)$, where, for $k=1,\dots,K$, $\Gamma_k$ is a full covariance matrix of size $r_k \times r_k$, associated with the $k$-th sub-group of correlated random effects. Let us now assume that we want to test that the $K$-th block of variances is null. Then, it can similarly be shown that 
$$T(\Theta,\theta^*)\cap T(\Theta_0,\theta^*)^{\perp} = \{0\}^p  \times \left(\bigotimes_{k=1}^{K-1} \{0\}^{r_k(r_k+1)/2} \right) \times \mathbb{S}_+^{r_K}  \times \{0\}^{\frac{J(J+1)}{2}}.$$

Moreover, thanks to the expressions of the cone $T(\Theta,\theta^*)\cap T(\Theta_0,\theta^*)^{\perp}$ established in Proposition \ref{prop:cone}, we can deduce that several weights involved in the chi-bar-square distribution defined in \eqref{eq:lrtdist} are equal to 0. The following corollary details this result for the two cases described in Proposition \ref{prop:cone}.

\begin{cor}
$\quad$
\begin{enumerate}
	\item Assume that $\Theta = \{\theta \in  \mathbb{R}^q \mid \beta \in \mathbb{R}^p, \Gamma \in \mathbb{S}_+^{p}, \Gamma \ \text{full}, 
\covmatrixerr \in \mathbb{S}^{\nbMeasPerInd}_+  \}$. 
	Then the distribution of the random variable $\bar{\chi}^2(I^{-1}_*,T(\Theta,\theta^*)\cap T(\Theta_0,\theta^*)^{\perp})$ is a mixture of $\left(\frac{r(r+1)}{2} + 1\right)$ chi-square distributions with degrees of freedom between $r(p-r)$ and $\left(r(p-r) + \frac{r(r+1)}{2}\right)$.
	\item Assume that $\Theta = \{\theta \in  \mathbb{R}^q \mid \beta \in \mathbb{R}^p, \Gamma \in \mathbb{S}_+^{p}, \Gamma \ \text{diagonal}, 
\covmatrixerr \in \mathbb{S}^{\nbMeasPerInd}_+  \}$. Then the distribution of the random variable $\bar{\chi}^2(I^{-1}_*,T(\Theta,\theta^*)\cap T(\Theta_0,\theta^*)^{\perp})$ is a mixture of $(r+1)$ chi-square distributions with degrees of freedom  between 0 and $r$.
\end{enumerate}
\end{cor}

\begin{proof}
Let $V$ be a positive-definite matrix and $\mathcal{C}$ a closed convex cone of $\mathbb{R}^q$. We denote by $\mathcal{C}^o = \{x \in \mathbb{R}^q \mid x^t y \leq 0, \ \forall y \in \mathcal{C} \}$ its polar cone.
We recall the following properties for the weights of the chi-bar-square distribution $\bar{\chi}^2(V,\mathcal{C})$ \citep{Sha85,Sha88}: 
\begin{enumerate}
	\item[(i)] for $0 \leq i \leq q$, $w_i(q,V,\mathcal{C}) = w_{q-i}(q,V,\mathcal{C}^o)$.
	\item[(ii)] if $\mathcal{C}$ is included in a linear space of dimension $(q-k)$, for $1 \leq k \leq q$, then the first $k$ weights $\{w_i(q,V ,\mathcal{C}^o), i=0, \dots, k-1\}$ are zero,
	\item[(iii)] if $\mathcal{C}$ contains a linear space of dimension $l$, for $1 \leq l \leq q$, then the last $l$ weights $\{w_i(q,V,\mathcal{C}^o), i=q-l+1, \dots, q\}$ are zero.
\end{enumerate}

In our case, $\mathcal{C} = T(\Theta,\theta^*)\cap T(\Theta_0,\theta^*)^{\perp}$ and $V = I_*^{-1}$, and we have for both cases mentioned in the corollary:
\begin{enumerate}
	\item $\mathcal{C} = \{0\}^p  \times \{0\}^\frac{(p-r)(p-r+1)}{2}  \times \mathbb{R}^{r(p-r)} \times \mathbb{S}_+^{r}  \times \{0\}^{\frac{J(J+1)}{2}}$ which is included in $\mathbb{R}^{r(p-r) + \frac{r(r+1)}{2}}$, i.e. in a linear space of dimension $q - \left(p + \frac{(p-r)(p-r+1)}{2} + \frac{J(J+1)}{2}\right)$. Therefore, using properties (i) and (ii) above, the weights $w_i(q,I^{-1}_*,\mathcal{C})$, for $i= r(p-r) + \frac{r(r+1)}{2}+ 1, \dots, q$ are zero. Moreover, $\mathcal{C}$ contains $\mathbb{R}^{r(p-r)}$, i.e. a linear space of dimension $r(p-r)$, which means using properties (i) and (iii) above, that the weights $w_i(q,I^{-1}_*,\mathcal{C})$, for $i=0, \dots, r(p-r)-1$ are zero.
	\item $\mathcal{C} = \{0\}^p  \times \{0\}^{p-r} \times \mathbb{R}_+^{r}  \times \{0\}^{\frac{J(J+1)}{2}}$ which is included in $\mathbb{R}^r$, a linear space of dimension $q-p-(p-r) - \frac{J(J+1)}{2}$. It follows that the weights $w_i(q,I^{-1}_*,\mathcal{C})$, for $i= r+1, \dots, q$ are zero. Then, since $\mathcal{C}$ does not contain any linear space of dimension $k>0$, all the other weigths are non-zero. 
\end{enumerate}
\end{proof}

\begin{rem}
The theoretical results above extend in a natural way to  the model defined in Section \ref{sec:NLMEmodel}. For example, it is possible to consider covariates depending on the individual in the model, leading to non identically distributed observations, provided that additional suitable assumptions are fulfilled \citep{Silv94}. It is also possible to consider more general models for the random effects and the error term.
\end{rem}

\section{Practical implementation}\label{sec:practical}

\subsection{Computation of the likelihood ratio test statistic}\label{sec:lrt}

The computation of the likelihood ratio test requires the computation of the maximum likelihood values under the null and alternative hypotheses, denoted respectively by $\hat{\theta}_0$ and $\hat{\theta}_1$, as well as the values of the likelihood at these two points, $L(y_1^\nbInd;\hat{\theta}_0)$ and $L(y_1^\nbInd;\hat{\theta}_1)$. 

However, in the context of nonlinear mixed effects models, the likelihood is not available in a closed form, and we need to resort to stochastic variants of the Expectation-Maximization (EM) algorithm, such as the Stochastic Approximation EM algorithm for example \citep{Kuh05}, to compute $\hat{\theta}_0$ and $\hat{\theta}_1$. For the same reason, $L(y_1^\nbInd;\hat{\theta}_0)$ and $L(y_1^\nbInd;\hat{\theta}_1)$ cannot be computed explicitly, and should be approximated using appropriate methods such as numerical or stochastic integration.

Since the decision to reject the null hypothesis relies on the value of the test statistics, the approximation of $L(y_1^\nbInd;\theta)$ must be computed precisely.
Let us denote by $L(\obs_i;\theta)$ the marginal likelihood of the $i$-th individual, and by $\ell(\obs_1^\nbInd;\theta)$ the joint log-likelihood. Then:
\begin{equation}\label{eq:lik}
 \ell(\obs_1^\nbInd;\theta)  = \log \left( \prod_{i=1}^\nbInd L(y_i;\theta) \right)  =  \sum_{i=1}^\nbInd \log \left(\int_{\mathbb{R}^\nbMiss} f(\obs_i \mid \missing_i ; \theta) p(\missing_i;\theta) d\missing_i \right),
\end{equation}
where $f(\cdot \mid \missing_i ; \theta)$ is the conditionnal probability density function  of $y_i$ given the random effect $ \varphi_i$, and $p(\cdot;\theta)$ is the probability density function of the random effect $\varphi_i$. This quantity can be approximated using classical methods for integral approximations. However in the case of high dimensional random effect, stochastic integration is preferred over numerical approximations, allowing a better approximation in the case of comparable computation times.

In practice, each $L(y_i;\theta)$ can be approximated independently of the others using Monte Carlo methods, and this can be done in parallel to optimize the execution time.
Computing the joint log-likelihood using the sum of the marginal log-likelihoods, instead of taking the logarithm of their product can also result in less numerical issues. %

To further reduce the variability of the LRT statistics estimate, we can compute directly the Monte Carlo estimate of $LRT_N$ based on the same sample size $M$ rather than using two estimates of $\ell(y_1^\nbInd;\hat{\theta}_0)$ and $\ell(y_1^\nbInd; \hat{\theta}_1)$ based on two different samples. Thus, let us consider:
\begin{equation}\label{eq:loglikIS}
\hat{LRT}_{N,M} = -2 \ \sum_{i=1}^\nbInd \log \frac{\sum_{m=1}^{M} f(\obs_i \mid \missing_{i,0}^m ; \hat{\theta}_0)}{\sum_{m=1}^{M} f(\obs_i \mid \missing_{i,1}^m ; \hat{\theta}_1) },
\end{equation}
where $\missing_{i,0}^m = \hat{\beta}_0 + \hat{\Gamma}_0^{1/2} Z_i^m$, $\missing_{i,1}^m = \hat{\beta}_1 + \hat{\Gamma}_1^{1/2} Z_i^m$ and $Z_i^m \sim \mathcal{N}(0,I_p)$.

The sample size $M$ of the Monte Carlo algorithm should be chosen large enough in order to ensure that the variance of the final estimate $\hat{LRT}_{N,M}$ is below a chosen threshold. 

\subsection{Computation of chi-bar-square weights when $\Gamma$ is diagonal}\label{sec:weights}
In general, the weights involved in the definition of the chi-bar-square distribution defined in equation \eqref{eq:chibarsquareFDR} are not available in a tractable form. However, in the special case of a diagonal covariance matrix, the cone $\mathcal{C}$ involved in the chi-bar-square distribution is polyhedral of dimension $r$ (see Proposition \ref{prop:cone}). Indeed, in this case the cone can be written as $\mathcal{C} = \{ \theta \in\mathbb{R}^q \mid R\theta \geq 0\}$, with $R= \left(\mathbf{0}_{r \times (p+p-r) } \mid I_r  \mid \mathbf{0}_{r \times \frac{J(J+1)}{2}}\right)$, a full-rank matrix of dimension $r \times q$.
For polyhedral cones of this type, \cite{Sha85} provided the exact weights expressions for $1 \leq r \leq 3$. 
The case $r =2$ is also treated by \cite{Self87}.

Following the notation of \cite{Sha85}, we denote by $\rho_{ij}= v_{ij}/(v_{ii} v_{jj})^{1/2}$ and $\rho_{ij.k} = (\rho_{ij} - \rho_{ik} \rho_{jk})/((1-\rho_{ik}^2)(1-\rho_{jk}^2))^{1/2}$, respectively the correlation coefficient, and the partial correlation coefficient associated with the covariance matrix $RI_*^{-1}R^t$, where $v_{ij}$ stands for the element in row $i$ and column $j$ of matrix $RI_*^{-1}R^t$. Using Proposition 3.6.1 of \cite{Sil11}, we have $w_{i}(q,I_*^{-1},\mathcal{C}) = w_i(r,RI^{-1}_*R^t,\mathbb{R}_+^r)$ and denoting by $w_{i,r} = w_i(r,RI^{-1}_*R^t,\mathbb{R}_+^r)$, we have the following expressions:

\begin{itemize}[leftmargin=*]
\item For $r=1$, we get $w_{0,1} = w_{1,1} = 1/2$.
\item For $r=2$, we have:  $w_{0,2} = 1/2 \ \pi^{-1} \cos^{-1}(\rho_{12})$, $w_{1,2} = 1/2$, and $w_{2,2} = 1/2 - 1/2 \ \pi^{-1} \cos^{-1}(\rho_{12})$
\item For $r=3$, we have: $w_{3,3} = (4\pi)^{-1} (2\pi - \cos^{-1}(\rho_{12}) - \cos^{-1}(\rho_{13}) - \cos^{-1}(\rho_{23}))$, $w_{2,3} = (4\pi)^{-1} (3\pi - \cos^{-1}(\rho_{12.3}) - \cos^{-1}(\rho_{13.2}) - \cos^{-1}(\rho_{23.1}))$, $w_{1,3} = 1/2 - w_{3,3}$, and $w_{0,3} = 1/2 - w_{2,3}$.
\end{itemize}

For $r>3$, and in more general settings, e.g. when $\mathcal{C}$ is not a polyhedral cone, one has to either approximate the weights through numerical integration or Monte Carlo simulations, or to directly compute the tail probability of the chi-bar-square distribution (see \cite{Sil11}, page 78). 
For spherical and polyhedral cones, \cite{Del07} proposed an elegant method based on Rice's formula.

\section{Experiments}\label{sec:numerical}

\subsection{Simulation study}

\subsubsection{Simulation settings and practical implementation}\label{sec:simsettings}
Let us consider the general mixed effects model presented in Section \ref{sec:defNLME}, for a set of observations $y_{ij}$, $i=1, \dots, \nbInd$, $j=1, \ldots J$ and random effects $\varphi_i $:
\begin{equation}\label{eq:orangeModel}
y_{ij} = g(\varphi_i,x_j) + \varepsilon_{ij} \ ,
\end{equation}
where  $\varphi_{i}\sim\mathcal {N}(\beta,\Gamma)$, $\varepsilon_{ij}\sim\mathcal {N}(0,\sigma^{2})$ and the $(\varepsilon_{ij})$ are independent. 

We consider a linear mixed effects model and the logistic mixed effects model described in Section \ref{sec:examples} with two and three random effects  to explore different settings. We denote by $\theta^*$ the true parameter value used to generate the data under $H_0$. \\
 Let us denote by $\mathcal{M}_1$ the linear model with three random effects where we set $g(\varphi_i,x_j) = \beta_1 + \varphi_{i1} + (\beta_2 + \varphi_{i2})x_j + (\beta_3 + \varphi_{i3}) x_j^2$. We choose $\beta^* = (0, 7, 2)^t$, $\gamma_1^* = 1.3$, $\gamma_2^* = 1$ and $\gamma_{12}^*= 1.04$, corresponding to a correlation coefficient of 0.8 between $\varphi_{i1}$ and $\varphi_{i2}$. We consider the null hypothesis $H_0$ defined by $\gamma_{3}^*=\gamma_{13}^*=\gamma_{23}^*=0$. In the sub-model with two random effects, we set $g(\varphi_i,x_j) = \beta_1 + \varphi_{i1} + (\beta_2 + \varphi_{i2}) x_j $, $\beta^* = (0, 7)^t$ and $\gamma_1^* = 1.3$. In this case we consider $H_0$ defined by  $\gamma_{2}^*=\gamma_{12}^*=0$. In each simulation settings, $x_j = j$ and $\sigma$ is chosen equal to $1.5$.

 Let us denote by $\mathcal{M}_2$ the logistic model with three random effects where we set  $g(\varphi_{i},x_j) = \frac{\varphi_{i1}}{1+\exp\left(-\frac{x_{j}-\varphi_{i2}}{\varphi_{i3}}\right)}$. We set $\beta^* = (200,500,150)^t$, $\gamma^*_2 =50$, $\gamma^*_{3} = 15$ and $\gamma^*_{23} = 375$, corresponding to a correlation coefficient of $0.5$ between $\varphi_{i2}$ and $\varphi_{i3}$. We consider here  the null hypothesis $H_0$ defined by $\gamma_{1}^*=\gamma_{12}^*=\gamma_{13}^*=0$. In the sub-model with two random effects, we set $\beta^* = (200,500)^t$ and $\gamma^*_2 = 50$. In this case we consider $H_0$ defined by  $\gamma_{1}^*=\gamma_{12}^*=0$. The vector of observation times $(x_{j})$ is the same for all the individuals, and is defined as a vector of 20 equally spaced values between 50 and 1000, plus 5 equally spaced values between 1100 and 1500. In each simulation settings, $\sigma$ is chosen equal to 10.
When only two random effects are considered in the model, $\beta_3$ is fixed to 150 and not estimated by the algorithm.

We consider several test cases.
For each, to evaluate the level of the test we generate $K$ datasets $D_{0,1},\dots,D_{0,K}$ under the null hypothesis, and we denote by $\hat{\theta}_ {0,k}$ (resp. $\hat{\theta}_{1,k}$) the maximum likelihood estimates of $\theta^*$ using the dataset $D_{0,k}$ under $H_0$ (resp. $H_1$). 
The likelihood ratio test statistics estimate is denoted by $\hat{LRT}_k$.
Then, the empirical level of the test for a sample size $K$ is equal to:
\begin{equation}
	\hat{\alpha}_{K} = \frac{1}{K} \sum_{k=1}^K \mathbb{1}_{\hat{LRT}_k > c_{\alpha}},
\end{equation}
where $c_{\alpha}$ is the $(1-\alpha)$ quantile of the limiting distribution of the LRT statistics. In practice, $c_{\alpha}$ is not always available in a closed form and may be estimated as mentioned in Section \ref{sec:practical}.

Parameter estimation was performed either using the \texttt{lmer} function implemented in the R package \texttt{lme4} \citep{lme4}, for the linear model, or using the SAEM algorithm implemented in the R package \texttt{saemix} \citep{Com11}, for the nonlinear model. Others parts of the codes were also developed in R.  

Note that in the linear mixed model case, the Fisher information matrix $I_*$ is known and is given by:
\begin{equation}\label{eq:fimLinear}
(I_*)_{i,j} = \left(\frac{\partial X\beta^*}{\partial \theta_i}\right)^t  \Omega^{-1} \ \frac{\partial X\beta^*}{\partial \theta_j} + \frac{1}{2} \text{Tr} \left(\Omega^{-1} \frac{\partial \Omega}{\partial \theta_i} \Omega^{-1} \frac{\partial \Omega}{\partial \theta_j} \right),
\end{equation}
where $\Omega = Z\Gamma^* Z^t + (\sigma^*)^2 I_J$, $\theta_i$ is the $i$-th element of vector $\theta$,  $\text{Tr}(A)$ denotes the trace of matrix $A$, for any matrix $A$, and where for a matrix $A$ of size $m \times n$, $\frac{\partial A}{\partial x}$ is the matrix of size $m \times n$ whose element $(i,j)$ is given by $\left( \frac{\partial A}{\partial x} \right)_{i,j} = \frac{\partial A_{ij}}{\partial x}$.

\subsubsection{Case studies and results}\label{sec:testcases}
For the two models ${\cal M}_1$ and  ${\cal M}_2$, we will consider five test cases as follows:
\begin{itemize}
	\item[\textsc{Case} 1:] Testing that one variance is zero in a model with two independent random effects.
In this case, the limiting distribution of the LRT is the mixture $0.5\chi_0^2 + 0.5 \chi_1^2$. 
	\item[\textsc{Case} 2:]  Testing that one variance is zero in a model with two non independent random effects.
In this case, the limiting distribution is the mixture $0.5 \chi_1^2 + 0.5 \chi_2^2$. 
	\item[\textsc{Case} 3:]  Testing that one variance is zero in a model with three independent random effects.
Here, the limiting distribution is the mixture $0.5 \chi_0^2 + 0.5 \chi_1^2$.
	\item[\textsc{Case} 4:]  Testing that one variance is zero in a model with three non independent random effects. In this case, the limiting distribution is the mixture $0.5 \chi_2^2 + 0.5 \chi_3^2$.
	\item[\textsc{Case} 5:]  Testing that two variances are zero in a model with three independent random effects.
Here, the limiting distribution is the mixture $w_{0,2} \chi_0^2 + 0.5 \chi_1^2 + (0.5-w_{0,2}) \chi_2^2$ (see Section \ref{sec:weights}).
\end{itemize}

Note that the limiting distribution is the same whatever the linear or nonlinear structure of the model. However it depends strongly on the correlation structure of the random effects.

We first analyzed the finite sample size properties of the LRT statistics when performing the test in the linear model $\mathcal{M}_1$ involving  two or three random effects, with and without correlations between the random effects. We started by testing that the variance of one random effect is zero. We computed the empirical level as detailed above  for nominal level $\alpha$ in $\{0.01,0.05,0.10\}$. The sample size $\nbInd$ varied in $\{100,500\}$. Results are presented in Tables \ref{tab:levelsLinearp2} and \ref{tab:levelsLinearp3r1}. We observe that the empirical levels are closer to nominal ones when $\nbInd$ grows, for random effects involving two and three components, whatever the correlation between the random effects. 

\begin{table}[h]
\centering
\caption{Empirical level for a given theoretical level $\alpha \in \{0.01,0.05,0.10\}$, when testing that one variance is zero in the \textbf{linear model} $\mathcal{M}_1$ with two random effects associated with a covariance matrix $\Gamma$ which is either full or diagonal, evaluated on $K=10000$ datasets of size $\nbInd  \in \{100,500\}$}
\label{tab:levelsLinearp2}
\begin{tabular}{|c|cc|cc|}
\hline
\multirow{2}{*}{$\alpha$} &  \multicolumn{2}{c|}{$\Gamma$ diagonal} & \multicolumn{2}{c|}{$\Gamma$ full} \\
\cline{2-5}
 & $N=100$  & $N=500$  & $N=100$  & $N=500$   \\
\hline
0.01 & 0.008  & 0.008  & 0.008 &  0.009 \\
0.05 & 0.040  & 0.044  & 0.040 &  0.045  \\
 0.10 & 0.084  & 0.089 & 0.085 &  0.092  \\
\hline
\end{tabular}
\end{table}

\begin{table}[h]
\centering
\caption{Empirical level for a given  theoretical level $\alpha  \in \{0.01,0.05,0.10\}$, when testing that one variance is null in the \textbf{linear model} $\mathcal{M}_1$ with three random effects associated with a covariance matrix $\Gamma$ which is either full or diagonal, evaluated on $K=10000$ datasets of size $\nbInd   \in \{100,500\}$}
\label{tab:levelsLinearp3r1}
\begin{tabular}{|c|cc|cc|}
\hline
\multirow{2}{*}{$\alpha$} &  \multicolumn{2}{c|}{$\Gamma$ diagonal} & \multicolumn{2}{c|}{$\Gamma$ full}  \\
\cline{2-5}
 & $N=100$  & $N=500$  & $N=100$ &  $N=500$ \\
\hline
0.01 & 0.008   & 0.009  & 0.007 &  0.008  \\
0.05 &  0.045  & 0.049  & 0.040 &  0.044  \\
0.10 &  0.088  & 0.097  & 0.082 &  0.093  \\ 
\hline
\end{tabular}
\end{table}

Let us now highlight that one can be led to false conclusions when performing the LRT test in a model without taking into account the presence of correlations between random effects. Indeed, we considered the test of one variance equal to zero in model $\mathcal{M}_1$ with a correlation between the two random effects. We have computed the empirical quantiles corresponding to the limiting distribution obtained when assuming no correlation between the randon effects, namely the distribution $0.5\chi_0^2+0.5 \chi_1^2$.  The empirical levels were computed for any nominal level $\alpha$ in $\{0.01,0.05,0.10\}$ for a sample size $\nbInd$ equal to $500$. Results are displayed in Table  \ref{tab:test2corrquantind}. We observed that the empirical levels in column $3$ are too large, leading to possibly wrong conclusions. This emphasizes that the presence of correlations between the random effects in the model plays a crucial role when performing a test on variance components in mixed effects models.

\begin{table}[htb]
\centering
\caption{Empirical level for a given  theoretical level $\alpha  \in \{0.01,0.05,0.10\}$, when testing that one variance is zero in the \textbf{linear model} $\mathcal{M}_1$ with two correlated random effects, using the theoretical quantiles of the limiting distribution $0.5\chi_1^2+0.5 \chi_2^2$ (column 2), and using the quantiles of the limiting distribution $0.5\chi_0^2+0.5 \chi_1^2$ obtained when considering uncorrelated random effects (column 3) for $\nbInd = 500$.}
\label{tab:test2corrquantind}
\begin{tabular}{|c|c|c|}
\hline
$\alpha$&$ \hat{\alpha}_{0.5\chi_1^2+0.5 \chi_2^2}$ &  $\hat{\alpha}_{0.5\chi_0^2+0.5 \chi_1^2}$\\
\hline
0.01 & 0.009 & 0.050 \\
0.05 & 0.045 & 0.174 \\
0.10 & 0.092 & 0.311 \\
 \hline
\end{tabular}
\end{table}

We then evaluated the finite sample size properties of the LRT statistics when performing the test of two variances equal to zero in the linear model $\mathcal{M}_1$ involving three independent random effects.  The empirical levels were computed for any nominal level $\alpha$ in $\{0.01,0.05,0.10\}$, for a sample size $\nbInd$ variyng in $\{100,500, 800\}$. The results are detailed in Table  \ref{tab:levelsLinearp3r2}. As previously, we observe that the empirical levels converge toward the nominal ones when $\nbInd$ grows. However the asymptotic effect seems to occur slower when testing that two variances are equal to zero. This may be explained by the fact that more parameters have to be estimated under the null hypothesis than in the case where one variance is tested equal to zero.

\begin{table}[h]
\centering
\caption{Empirical level for a given theoretical level $\alpha \in \{0.01,0.05,0.10\}$, when testing that the two variances $\gamma_1^2$ and $\gamma_2^2$ are zero, in the \textbf{linear model} $\mathcal{M}_1$  with three random effects associated with a diagonal covariance matrix $\Gamma$, evaluated on $K=10000$ datasets of size $\nbInd   \in \{100,500,800\}$}
\label{tab:levelsLinearp3r2}
\begin{tabular}{|c|ccc|}
\hline
$\alpha$ & $N=100$  & $N=500$ & $N=800$ \\
\hline
0.01 & 0.007 & 0.007 & 0.010  \\
0.05 & 0.036  & 0.045 & 0.048  \\
0.10 & 0.074  & 0.090 & 0.096 \\ 
\hline
\end{tabular}
\end{table}

Let us now focus on the nonlinear model $\mathcal{M}_2$. We analyzed  the finite sample size properties of the LRT statistics when performing the test  in this model involving  two and three random effects, with and without correlations between random effects. The empirical levels were computed for any nominal level $\alpha$ in $\{0.01,0.05,0.10\}$ and a sample size $\nbInd$ equal to $500$. The results are detailed in Table  \ref{tab:levelsnonlinear}. We observed that, in the case of this nonlinear mixed effects model, the empirical levels are not as close to the nominal ones as in the linear case, and are lower than the nominal levels. We studied the asymptotic behaviour of the empirical levels and noticed that they were not improved when $\nbInd$ was increased to 1000 (results not presented). Due to the computational time involved, we did not test higher values of $\nbInd$. These numerical results might be explained by the numerical integrations which have to be performed to compute the LRT statistics in nonlinear mixed effects models, which is not the case in the linear setting.  However the empirical levels in the considered model with presence of correlation between the random effects were globally slightly lower than those in the models without correlations. This may be explained again by the fact that the models with correlations involved more parameters to estimate. 
 We also observed that the empirical levels in model $\mathcal{M}_2$ with $3$ random effects are lower than those in model $\mathcal{M}_2$ involving only $2$ random effects. The same argument can be retained in this case.

\begin{table}[h]
\centering
\caption{Empirical level for a given theoretical level $\alpha \in \{0.01,0.05,0.10\}$, when testing that one variance is zero, in the \textbf{nonlinear growth curve model} $\mathcal{M}_2$ with $\nbMiss$ random effects, associated with a covariance matrix $\Gamma$ which is either full or diagonal, evaluated on $K=1000$ datasets of size $\nbInd = 500$}
\label{tab:levelsnonlinear}
\begin{tabular}{|c|c|c|c|c|}
\hline
\multirow{1}{*}{$\alpha$} & \multicolumn{2}{c|}{$\nbMiss=2$} & \multicolumn{2}{c|}{$\nbMiss = 3$} \\
\cline{2-5}
  & diagonal & full & diagonal & full \\
 0.01 &  0.003   & 0.007 & 0 & 0.003 \\
0.05 & 0.038 &  0.033 & 0.040  & 0.033\\
0.10 & 0.082 &  0.073 & 0.077 & 0.073\\
\hline
\end{tabular}
\end{table}

Finally we  assessed the empirical power of the procedure when performing the test of one variance equal to zero in  the  model $\mathcal{M}_2$ without correlation between the two random effects. We computed the empirical power for a nominal level $\alpha$ equal to $0.05$ and for  $\nbInd$ equal to $500$. For $\gamma_1$ in $\{2,5,7\}$, we obtain respectively $\{0.23, 0.57, 0.98\}$.
 The procedure seemed to have good properties for detecting alternatives even though we have no theoretical results to assess this observation.

\subsection{Real data analysis}
The method was illustrated on two sets of real data. The first one is the famous dental growth dataset from \cite{Pott64}, in which the distance from the center of the pituitary gland to the pteryomaxillary fissure was measured at 4 different ages for 27 children (16 boys and 11 girls). This dataset is available  in the \texttt{R} package \texttt{mice}.
The second dataset comes from a study of coucal growth rates, available as a Dryad package \citep{dryad}. Body masses of 678 nestlings from two species (white-browed coucals and black coucals) were recorded every two days from their hatching date until they left the nests. In this paper, we only consider data from the white-browed coucals species, corresponding to the highest sample size (385 individuals).

\subsubsection{Dental growth data}
A linear model was fitted to the dental growth data using the \texttt{lme4} package, with two random effects as described in Section \ref{sec:simsettings} (model $\mathcal{M}_1$). More precisely, if we denote by $y_{ij}$, $1 \leq i \leq 27$, $1 \leq j \leq 4$, the dental measurement of child $i$ at age $x_j$, the following model was considered, with a random slope and a random intercept:
\begin{align}\label{eq:potthoff}
y_{ij} & = \beta_1 + \varphi_{i1} + (\beta_2 + \varphi_{i2}) x_j + \varepsilon_{ij}, \quad \varepsilon_{ij} \sim \mathcal{N}(0,\sigma^2) \\
\varphi_i & = (\varphi_{i1}, \varphi_{i2})^t \sim \mathcal{N}(0,\Gamma).
\end{align}

We tested whether the variance of the slope is equal to $0$. For that purpose we considered \textsc{Case} 1 and \textsc{Case} 2 described in Section \ref{sec:testcases}, according to the structure of $\Gamma$ one wishes to consider in the alternative hypothesis. In the first case, we considered a diagonal covariance matrix $\Gamma = \diag(\gamma_1^2, \gamma_2^2)$ and we tested $H_0 : \{ \gamma_1^2 = 0, \gamma_2^2 \geq 0 \}$ against $H_1 : \{ \gamma_1^2 \geq 0, \gamma_2^2 \geq 0 \}$. In the second case we considered a full covariance matrix $\Gamma = (\gamma_1^2 \ \gamma_{12} \mid \gamma_{12} \ \gamma_2^2)$ and we tested $H_0 : \{ \gamma_1^2 = \gamma_{12} = 0, \gamma_2^2 \geq 0 \}$ against $H_1 : \{ \Gamma \geq 0 \}$.
In both cases, we computed the likelihood ratio test and compared it to the rejection threshold $q_{\alpha}^{d}$ associated to the limiting distribution $d$ and corresponding to an asymptotic level  $\alpha$. 
We recall that in \textsc{Case} 1 the limiting distribution of the test statistics is the mixture $0.5\chi_0^2 + 0.5 \chi_1^2$, while in \textsc{Case} 2 the limiting distribution of the test statistics is the mixture $0.5\chi_1^2 + 0.5 \chi_2^2$.
We also computed the $p$-value associated to the test.

In the case where a diagonal structure is assumed for $\Gamma$, the test statistics is equal to $LRT_{diag} = 3.651$, the rejection threshold is equal to $q^{0.5\chi_0^2 + 0.5 \chi_1^2}_{0.05} = 2.706$ and the $p$-value is equal to 0.028. On the other hand, when a full covariance structure is assumed for $\Gamma$, the test statistics is equal to $LRT_{full} = 4.178$, the rejection threshold is equal to $q^{0.5\chi_1^2 + 0.5 \chi_2^2}_{0.05} = 5.139$ and the $p$-value is equal to 0.082. 
Such result highlights in particular that depending on the structure of $\Gamma$, the conclusion regarding the null hypothesis can be different depending on the level of the test. In our test, at the asymptotic level of 5\%, if we assume that there is no correlation between the two random effects, we reject the null hypothesis, but if we assume that the two random effects are correlated, we do not reject the null hypothesis.

\subsubsection{Growth rate of coucals}

A nonlinear model was fitted to the white-browed coucals data using the \texttt{saemix} package, with three random effects as described in Section \ref{sec:simsettings} (model $\mathcal{M}_2$). More precisely, if we denote by $y_{ij}$, $1 \leq i \leq 385$, $1 \leq j \leq n_i$ the body mass of nestling $i$ at age $x_j$, we considered the following model:
\begin{equation}\label{eq:coucal}
  y_{ij}  =\frac{\varphi_{i1}}{1+\exp\left(\frac{x_{j}-\varphi_{i2}}{\varphi_{i3}}\right)}, \quad \varepsilon_{ij} \sim \mathcal{N}(0,\sigma^2), \quad \varphi_i  \sim \mathcal{N}(\beta,\Gamma),
\end{equation}
where $\varphi_{i1}$ is the asymptotic body mass of individual $i$, $\varphi_{i2}$ the age in days at which individual $i$ reaches half its asymptotic body mass and $\varphi_{i3}$ the growth rate of individual $i$.

We tested whether the variances of the inflexion point and the growth rate were equal to 0, and thus we considered \textsc{Case} 4 described in Section \ref{sec:testcases}. We considered a diagonal covariance matrix $\Gamma = \diag(\gamma_1^2, \gamma_2^2, \gamma_3^2)$ and we tested $H_0 : \{\gamma_1^2 \geq 0, \gamma_2^2 = \gamma_3^2 = 0\}$ against $H_1 : \{\gamma_1^2 \geq 0, \gamma_2^2 \geq 0, \gamma_3^2 \geq 0\}$. The limiting distribution of the LRT statistics is the mixture $w_{0,2} \chi_0^2 + 0.5 \chi_1^2 + (0.5-w_{0,2}) \chi_2^2$, where $w_{0,2}$ is defined in Section \ref{sec:weights} and can be computed from the correlation coefficient between parameters $\gamma_2^2$ and $\gamma_3^2$ obtained from the Fisher information matrix. 

The estimated Fisher information matrix $\hat{I}$ is obtained as an output of the \texttt{saemix} package when the full model corresponding to $H_1$ is fitted to the data. From this matrix we can easily compute the covariance matrix $\hat{V} =  R\hat{I}^{-1} R^t$, where $R = (\mathbf{0}_{2\times 4}  \ | \ I_2 \ | \ \mathbf{0}_{2\times1})$, and the corresponding correlation matrix $\hat{C} = \text{diag}(\hat{V} )^{-\frac{1}{2}} \, \hat{V} \, \text{diag}(\hat{V} )^{-\frac{1}{2}}$. Then, the correlation coefficient $\hat{\rho}_{12}$ needed to compute $w_{0,2}$, is the element $(1,2)$ of $\hat{C} $.

In our case, $\hat{\rho}_{12} = -0.644$, leading to the three following weights: $w_{0,2} = 0.139$, $w_{1,2} = 0.5$ and $w_{2,2} = 0.361$, and thus to the limiting distribution $0.139 \chi_0^2 + 0.5 \chi_1^2 + 0.361 \chi_2^2$. The likelihood ratio test statistics is equal to $LRT=3.119$, and the rejection threshold is equal to $q^{0.139 \chi_0^2 + 0.5 \chi_1^2 + 0.361 \chi_2^2}_{0.05} = 4.682$. 
The corresponding $p$-value is evaluated at 0.114 on this dataset. In other words, assuming a diagonal covariance matrix $\Gamma$, we do not reject the null hypothesis that both the inflection point and the growth rates are fixed effects and do not vary among individuals, at the asymptotic level of 5\%.

\section{Discussion}\label{sec:discussion}

We considered in this paper the likelihood ratio test for testing variance components in general  mixed effects models, including nonlinear ones,  and established that its asymptotic distribution is a chi-bar-square one. We also identified the corresponding weights.
We highlighted in particular that this distribution  depends on the presence of correlations between the random effects.  
We also provided practical guidelines for the computation of the test statistics. We carried out a simulation study on a linear mixed effects model and on a  nonlinear growth curve model to illustrate the finite sample size properties of the procedure. The simulations also showed the impact of considering that the random effects are independent or not in the model considered.

Several perspectives are of great interest from a theoretical  as well as from a practical point of view. For example, the issue of variance components testing in mixed effects models  has gained an increasing interest in plant growth modelling. Models of plant growth have raised expectations to help improve the understanding of gene by environment interactions by developing a predictive capacity that scales from genotype to phenotype \citep{letort2008}. In plant ecophysiological models, one genotype should be represented by one unique set of parameters, and reversely, two different genotypes should potentially be characterized by two different sets of parameters \citep{Tar03}. 
Such models are often descriptive ones, involving mechanistic parameters. Therefore, considering these parameters as random effects is relevant in order to understand how they vary within a given population \citep{Baey16} and statistical tools to identify fixed from random effects are necessary.  Therefore, adapted tools to compute precisely the proposed test statistics in nonlinear mixed effects models have to be developed. Indeed  the reliability of the test is linked to the precise evaluation of the test statistics. Moreover, the procedure opens very promising perspectives in applications with high dimensional random effects, in particular in genetics. However this would lead to many computational issues, in particular for the computation of the Fisher information matrix. Again, more advanced works on the computational methods are still necessary to allow tackling high-dimensional problems with confidence. Besides, from a theoretical point of view it would be interesting to establish results for the test power. Finally since the asymptotic regime is often not reached in practice, it would be of great  interest to develop a finite sample-size procedure using for example bootstrap methods or permutation tests in the spirit of the ones developed in the context of linear mixed effects models.

\section{Appendix}

We begin by recalling the definition of an approximating cone:
\begin{definition}\label{def:cone}
\citep{Cher54}. A set $\mathcal{A}$ is said to be an \emph{approximating cone} of a set $\Theta$ at $\theta_0 \in \Theta$ if 
\begin{align*}
d(\theta,\mathcal{A}) & = o(\parallel \theta - \theta_0 \parallel) \text{ for }\theta \in \Theta\\
d(a,\Theta) & = o(\parallel a - \theta_0 \parallel) \text{ for } a \in \mathcal{A}.
\end{align*}
\end{definition}

\subsection{Proof of Theorem \ref{th:lrt}}

\paragraph{Consistency of the MLE on $\Theta_0$ and $\Theta$}
The $\sqrt{\nbInd}$-consistency of the MLE when the true parameter value $\theta^*$ lies on the boundary of the parameter space is ensured by results from \cite{And99}.

\paragraph{Tangent and approximating cones}
The next step of the proof is to show that both $\Theta_0$ and $\Theta$ can be approximated by cones, on which the consistency of the MLE also holds.

The sets $\Theta_0$ and $\Theta$ are defined by:
\begin{align*}
	\Theta_0 & = \{\theta \in  \mathbb{R}^q \mid \beta \in \mathbb{R}^p, \Gamma_1 \in \mathbb{S}_+^{p-r}, \Gamma_{12} = 0, \Gamma_2 = 0, \Sigma \in \mathbb{S}^{\nbMeasPerInd}_+ \} \\
	\Theta & = \{\theta \in  \mathbb{R}^q \mid \beta \in \mathbb{R}^p, \Gamma \in \mathbb{S}_+^{p}, \Sigma \in \mathbb{S}^{\nbMeasPerInd}_+ \},
\end{align*}
where the constraints are that $\Gamma_1$, of size $(p-r) \times (p-r)$ in $\Theta_0$, and $\Gamma$, of size $p \times p$ in $\Theta$, are positive semi-definite.
Since a matrix is positive semi-definite if and only if all its leading principal minors are positive (Sylvester's criterion), both $\Theta_0$ and $\Theta$ can be written as a set of polynomial equalities and inequalities corresponding to the different determinants. In particular, it means that both $\Theta_0$ and $\Theta$ are semi-algebraic sets. They are therefore \textit{Chernoff-regular}, i.e. they admit approximating cones at every point, and in particular at every point $\theta^* \in \Theta_0$ \citep{Drton09}.

According to Definition \ref{def:cone}, it means that if a sequence $\hat{\theta}_\nbInd$ is a $\sqrt{\nbInd}$-consistent estimator for $\theta^*$, then the squared distance between $\hat{\theta}_\nbInd$ and each parameter space $\Theta_0$ and $\Theta$ is equal, up to a term of order $1/\nbInd$ in probability, to the squared distance between $\hat{\theta}_\nbInd$ and the corresponding approximating cone at $\theta^*$. In other words, around the true value, it is possible to substitute $\Theta_0$ and $\Theta$ by their approximating cones, ensuring the consistency of the MLE on the approximating cones of $\Theta_0$ and $\Theta$ at $\theta^*$.

Approximating and tangent cones are linked, due to \cite{Gey94}. More precisely, if we denote by $T(\Theta, \theta^*)$ the tangent cone of $\Theta$ at $\theta^*$ and by $\mathcal{A}(\Theta,\theta^*)$ the approximating cone of $\Theta$ at $\theta^*$ then $\mathcal{A}(\Theta,\theta^*) = \theta^* + T(\Theta, \theta^*)$. 
In particular, it means that it is also possible to substitute the sets $\Theta_0$ and $\Theta$ by their tangent cones, on which the consistency of the MLE also holds.

\paragraph{Quadratic expansion of log-likelihood function}
Under conditions \ref{cond:likelihood} and using Taylor series expansion (see \cite{Silv94} and \cite{And99}), we can derive a quadratic expansion of the log-likelihood function. For the cases where some components of $\theta^*$ are on the boundary of the parameter space, one can use directional derivatives for those components. Let $K > 0$ and let us denote by $I_\nbInd(\theta^*)$ the Fisher information matrix based on $\nbInd$ observations, i.e. $I_\nbInd (\theta^*) = N I_*$. Then, we have:
\begin{align}\label{eq:dvtvrais}
\nonumber	\ell_N(\theta)  & = \ \ell_\nbInd(\theta^*) + (\theta - \theta^*)^t S_\nbInd(\theta^*) - \frac{1}{2} (\theta - \theta^*)^t I_\nbInd (\theta^*) (\theta - \theta^*) + r_\nbInd (\theta - \theta^*), \\
\nonumber & = \ \ell_\nbInd(\theta^*) + \frac{1}{\sqrt{\nbInd}} u^t S_\nbInd (\theta^*) - \frac{1}{2} u^t I_* u + \tilde{r}_\nbInd (u), \\
	& = \ \ell_\nbInd(\theta^*) + \frac{1}{2\nbInd} S_{\nbInd} (\theta^*)^t I_*^{-1} S_{\nbInd} (\theta^*) - \frac{1}{2} [Z_{\nbInd} - u]^t I_* [Z_{\nbInd} - u] + \tilde{r}_\nbInd(u),
\end{align}
where  $\ell_N(\theta)$ stands for $\ell(y_1^N;\theta)$, $S_{\nbInd} (\theta)= (\partial/\partial \theta) \ell_N(\theta)$ is the score function, $Z_{\nbInd} =\nbInd^{-1/2} I_*^{-1} S_{\nbInd}(\theta^*)$, and where $\sup_{||u|| < K} |\tilde{r}_\nbInd (u) | = o_p(1).$ 

The asymptotic normality of the score function can be shown using the central limit theorem under appropriate conditions on the likelihood function $f$ which are satisfied under conditions \ref{cond:likelihood}. We then have that $\nbInd^{-1/2} S_{\nbInd} (\theta^*) \xrightarrow{d} \mathcal{N}(0,I_*)$, where the components of $S_{\nbInd}(\theta^*)$ and $I_*$ corresponding to the elements of $\theta^*$ that are on the boundary of the parameter space $\Theta$ are directional derivatives.
In particular, it implies that $Z_{\nbInd}  \xrightarrow{d} Z$, where $Z =  \mathcal{N}(0,I_*^{-1})$. 

Let us define $\| x \|_V := x^t V x$. Thanks to the $\sqrt{\nbInd}$- consistency of the MLE on both sets $\Theta_0$ and $\Theta$, it is possible to restrict ourselves to a neighborhood of the form $\{\theta \mid \sqrt{\nbInd} ||\theta - \theta^* || < K \}$ for some positive constant $K$, in which case the remainder term in equation \eqref{eq:dvtvrais} becomes negligible. 
Moreover, since the first two terms in \eqref{eq:dvtvrais} do not depend on $\theta$, the likelihood ratio test statistics can be written as:
\begin{align*}
LRT_{\nbInd} = & - 2 [ \sup_{\theta \in \Theta_0} \ell_{\nbInd}(\theta) - \sup_{\theta \in \Theta} \ell_{\nbInd}(\theta) ] \\
	= & \inf_{u \in \Theta_0} [Z_{\nbInd} - u]^t I_* [Z_{\nbInd} - u] - \inf_{u \in \Theta} [Z_{\nbInd} - u]^t I_* [Z_{\nbInd} - u] + o_p(1)\\
	= & \inf_{u \in \Theta_0} \|Z_{\nbInd} - u\|^2_{I_*}   - \inf_{u \in \Theta} \|Z_{\nbInd} - u\|^2_{I_*} + o_p(1).
\end{align*}

Now, we can substitute $\Theta_0$ and $\Theta$, around $\theta^*$, by their tangent cones denoted by $T(\Theta_0,\theta^*)$ and $T(\Theta,\theta^*)$ respectively. 
Recall that $Z_\nbInd = \nbInd^{-1/2} I_*^{-1} S_{\nbInd}(\theta^*)$ and $u = \sqrt{\nbInd} (\theta - \theta^*)$. Then, following the lines of the proof of \cite{Sil11}, we have that:
\begin{align}\label{eq:diffDistanceCones}
\nonumber	  LRT_\nbInd = & \inf_{\theta \in \Theta_0} \nbInd \|\nbInd^{-1/2}Z_{\nbInd} - (\theta - \theta^*)\|^2_{I_*}   - \inf_{\theta \in \Theta}  \nbInd \|\nbInd^{-1/2}Z_{\nbInd} - (\theta - \theta^*)\|^2_{I_*} + o_p(1) \\
\nonumber	 = & \ \nbInd \| \nbInd^{-1/2}Z_{\nbInd} + \theta^*- \Theta_0\|^2_{I_*}   -  \nbInd  \| \nbInd^{-1/2}Z_{\nbInd} + \theta^*- \Theta\|^2_{I_*} + o_p(1) \\
\nonumber = & \ \nbInd \| \nbInd^{-1/2}Z_{\nbInd} + \theta^*- \mathcal{A}(\Theta_0,\theta^*)\|^2_{I_*}   -  \nbInd  \| \nbInd^{-1/2}Z_{\nbInd} + \theta^*- \mathcal{A}(\Theta,\theta^*)\|^2_{I_*} + o_p(1) \\
\nonumber = & \ \| Z_{\nbInd} - T(\Theta_0,\theta^*)\|^2_{I_*}   -  \| Z_{\nbInd} - T(\Theta,\theta^*)\|^2_{I_*} + o_p(1) \\
					\xrightarrow{d} & \| Z - T(\Theta_0,\theta^*)\|^2_{I_*} -  \|Z - T(\Theta,\theta^*) \|^2_{I_*},
\end{align}
where the transition between lines 2 and 3 uses Corollary 4.7.5 in \cite{Sil11}.

In particular, the asymptotic distribution of $LRT_\nbInd$ is the distribution of the likelihood ratio test statistics for testing that the mean of a multivariate Gaussian distribution is in $T(\Theta_0,\theta^*)$, against the alternative that it is in $T(\Theta,\theta^*)$, based on one single observation $Z$.

\paragraph{Asymptotic distribution as a chi-bar-square}

Using Proposition \ref{prop:cone}, we can show that for covariance matrices diagonal or full,  $T(\Theta_0,\theta^*)$ is a linear space, which is contained in $T(\Theta,\theta^*)$, a closed convex cone. This can also be shown for general structures. Therefore, applying Theorem 3.7.1 of \cite{Sil11}, we can prove that the asymptotic distribution identified in \eqref{eq:diffDistanceCones} is equal to a chi-bar-square distribution, leading to the result:
\begin{equation*}
LRT_\nbInd 	\xrightarrow{d} \bar{\chi}^2(I^{-1}_*,T(\Theta,\theta^*)\cap T(\Theta_0,\theta^*)^{\perp}).
\end{equation*}

\hfill $\square$

\subsection{Proof of Proposition 1}

To calculate the tangent cones to $\Theta_0$ and $\Theta$ at $\theta^*$, we can use general results from \cite{Hir96} on the definition of tangent cones, and more recent results by \cite{Hir12} on the tangent cone of the set of symmetric positive semi-definite matrices.

We treat here the case where the covariance matrix $\Gamma$ is full, but similar tools can be used in the case where a more sparse structure is assumed for $\Gamma$.

\paragraph{Tangent cone of $\Theta$}
We recall that $\Theta$ is defined as:
\begin{align*}
\Theta & = \{\theta \in  \mathbb{R}^q \mid \beta \in \mathbb{R}^p, \Gamma \in \mathbb{S}_+^{p}, \covmatrixerr \in \mathbb{S}^{\nbMeasPerInd}_+  \} \\
\Theta & =  \mathbb{R}^p \times \mathbb{S}^p_+ \times \mathbb{S}^{\nbMeasPerInd}_+.
\end{align*}

Now, since each term in the above product is a convex cone, we have that $T( \mathbb{R}^p \times \mathbb{S}^p_+ \times \mathbb{S}^{\nbMeasPerInd}_+, (\beta^*,\Gamma^*,\Sigma^*)) =  T(\mathbb{R}^p,\beta^*) \times T(\mathbb{S}^p_+,\Gamma^*) \times T(\mathbb{S}^{\nbMeasPerInd}_+,\Sigma^*)$ (see for example \cite[Proposition 5.3.1.]{Hir96}).
Therefore, the tangent cone of $\Theta$ at $\theta^*$ is given by:
\begin{align*}
		T(\Theta,\theta^*) & = \mathbb{R}^p \times T(\mathbb{S}^p_+,\Gamma^*)  \times T(\mathbb{S}^{\nbMeasPerInd}_+,\Sigma^*),
\end{align*}
where $T(\mathbb{S}^p_+,\Gamma^*)$  is the tangent cone of the set of symmetric positive semi-definite matrices of size $p \times p$ at $\Gamma^*$, and $T(\mathbb{S}^{\nbMeasPerInd}_+,\Sigma^*)$ the tangent cone of the set of symmetric positive semi-definite matrices of size $J \times J$ at $\Sigma^*$.

To identify $T(\mathbb{S}^p_+,\Gamma^*)$ and $T(\mathbb{S}^{\nbMeasPerInd}_+,\Sigma^*)$, we can use the result established by \cite{Hir12}. According to the authors, the tangent cone of $\mathbb{S}_+^p$ at $A \in \mathbb{S}_+^p$ is given by $T_A = \{M \in \mathbb{S}^p \mid \langle Mu, u \rangle \geq 0 \text{ for all } u \in \ker A \}$, where $\mathbb{S}^p$ is the set of symmetric matrices of size $p \times p$.

In our case, since  $\Gamma^* =  \left[\begin{smallmatrix}
\Gamma_{1}^* & 0  \\
 0 & 0
\end{smallmatrix}\right]$, we have:
\begin{align*}
	T(\mathbb{S}^p_+,\Gamma^*) & =  \left\{ M \in \mathbb{S}^p \mid \forall u \in \ker \Gamma^*, \langle Mu, u \rangle \geq 0  \right\} \\
									& =  \left\{ \left( \begin{array}{c|c}
M_{11} &  M_{12} \\
\hline
M_{12}^t & M_{22}
\end{array}
\right) \in \mathbb{S}^p \mid \forall u = (\underbrace{0, \dots, 0}_{p-r}, u_{p-r+1}, \dots, u_p), u^t M u \geq 0 \right\} \\
									& =  \left\{ \left( \begin{array}{c|c}
M_{11} &  M_{12} \\
\hline
M_{12}^t & M_{22}
\end{array}
\right) \in \mathbb{S}^p \mid M_{22} \geq 0 \right\} \\
 T(\mathbb{S}^p_+,\Gamma^*) & = \mathbb{R}^{\frac{(p-r)(p-r+1)}{2}} \times \mathbb{R}^{r(p-r)} \times \mathbb{S}_+^{r}.
\end{align*}

Similarly, we have:
\begin{align*}
T(\mathbb{S}^{\nbMeasPerInd}_+,\Sigma^*) & = \left\{ M \in \mathbb{S}^J \mid \forall u \in \ker \Sigma^*, \langle Mu, u \rangle \geq 0  \right\} \\
											& = \left\{ M \in \mathbb{S}^J \mid \forall u \in \{0\}, \langle Mu, u \rangle \geq 0  \right\} \\
											& = \mathbb{S}^J
\end{align*}

In the end, we get:
\begin{align*}
		T(\Theta,\theta^*) & = \mathbb{R}^p \times \mathbb{R}^{\frac{(p-r)(p-r+1)}{2}} \times \mathbb{R}^{r(p-r)} \times \mathbb{S}_+^{r} \times \mathbb{S}^{\nbMeasPerInd}.
\end{align*}

In the case where $\Gamma$ is diagonal, i.e. when the effects are assumed to be independent, we have:
\begin{align*}
		T(\Theta,\theta^*) & = \mathbb{R}^p \times \mathbb{R}^{p-r} \times \mathbb{R}_+^{r} \times \mathbb{S}^{\nbMeasPerInd}.
\end{align*}	

\paragraph{Tangent cone to $\Theta_0$}
We recall that $\Theta_0$ is defined as:

\begin{align*}
	\Theta_0 = & \{\theta \in  \mathbb{R}^q \mid \beta \in \mathbb{R}^p, \Gamma_1 \in \mathbb{S}_+^{p-r}, \Gamma_{12} = 0, \Gamma_2 = 0, \covmatrixerr \in \mathbb{S}^{\nbMeasPerInd}_+  \} \\
	 = & \mathbb{R}^p \times \mathbb{S}_+^{p-r} \times \{0\}^{r(p-r)} \times \{0\}^{\frac{r(r+1)}{2}} \times \mathbb{S}_+^J.
\end{align*}

Using similar tools than for $\Theta$, we get:
\begin{align*}
T(\Theta_0, \theta^*) & = \mathbb{R}^p \times \mathbb{S}^{p-r} \times \{0\}^{r(p-r)} \times \{0\}^{\frac{r(r+1)}{2}} \times \mathbb{S}^J \\
 & = \mathbb{R}^p \times \mathbb{R}^\frac{(p-r)(p-r+1)}{2} \times \{0\}^{ r(p-r) + \frac{r(r+1)}{2}} \times \mathbb{R}^{\frac{J(J+1)}{2}}.
\end{align*}

Note that in the particular case where $\Gamma$ is diagonal, i.e. when the effects are supposed to be independent, the parameter space $\Theta_0$ and hence its tangent cone can be simplified, and we have:
\begin{align*}
		T(\Theta_0,\theta^*) & = \mathbb{R}^p \times \mathbb{R}^{p-r} \times \{0\}^r \times \mathbb{R}^{\frac{J(J+1)}{2}}.
\end{align*}	

\hfill $\square$

\bibliographystyle{abbrvnat}
\bibliography{biblio}

\begin{thebibliography}{38}
\providecommand{\natexlab}[1]{#1}
\providecommand{\url}[1]{\texttt{#1}}
\expandafter\ifx\csname urlstyle\endcsname\relax
  \providecommand{\doi}[1]{doi: #1}\else
  \providecommand{\doi}{doi: \begingroup \urlstyle{rm}\Url}\fi

\bibitem[Andrews(1999)]{And99}
D.~W. Andrews.
\newblock Estimation when a parameter is on a boundary.
\newblock \emph{Econometrica}, pages 1341--1383, 1999.

\bibitem[Baey et~al.(2016)Baey, Trevezas, and Courn\`ede]{Baey16}
C.~Baey, S.~Trevezas, and P.-H. Courn\`ede.
\newblock {A nonlinear mixed effects model of plant growth and estimation via
  stochastic variants of the EM algorithm}.
\newblock \emph{Communications in Statistics -- Theory and Methods},
  45\penalty0 (6):\penalty0 1643--1669, 2016.

\bibitem[Bates et~al.(2015)Bates, M{\"a}chler, Bolker, and Walker]{lme4}
D.~Bates, M.~M{\"a}chler, B.~Bolker, and S.~Walker.
\newblock Fitting linear mixed-effects models using {lme4}.
\newblock \emph{Journal of Statistical Software}, 67\penalty0 (1):\penalty0
  1--48, 2015.

\bibitem[Chant(1974)]{Chant74}
D.~Chant.
\newblock {On Asymptotic Tests of Composite Hypotheses in Nonstandard
  Conditions}.
\newblock \emph{Biometrika}, 61\penalty0 (2):\penalty0 291--298, 1974.

\bibitem[Chernoff(1954)]{Cher54}
H.~Chernoff.
\newblock {On the Distribution of the Likelihood Ratio}.
\newblock \emph{The Annals of Mathematical Statistics}, 25\penalty0
  (3):\penalty0 573--578, 1954.

\bibitem[Comets et~al.(2011)Comets, Lavenu, and Lavielle]{Com11}
E.~Comets, A.~Lavenu, and M.~Lavielle.
\newblock {SAEMIX}, an {R} version of the {SAEM} algorithm.
\newblock \emph{20$^{th}$ meeting of the {P}opulation {A}pproach {G}roup in
  {E}urope, Athens, Greece}, 2011.

\bibitem[Crainiceanu and Ruppert(2004)]{Crai04a}
C.~M. Crainiceanu and D.~Ruppert.
\newblock {Likelihood ratio tests in linear mixed models with one variance
  component}.
\newblock \emph{Journal of the Royal Statistical Society. Series B: Statistical
  Methodology}, 66\penalty0 (1):\penalty0 165--185, feb 2004.

\bibitem[Davidian and Giltinan(2003)]{Dav03}
M.~Davidian and D.~M. Giltinan.
\newblock {Nonlinear models for repeated measurement data: An overview and
  update}.
\newblock \emph{Journal of Agricultural, Biological, and Environmental
  Statistics}, 8\penalty0 (4):\penalty0 387--419, 2003.

\bibitem[Delmas and Foulley(2007)]{Del07}
C.~Delmas and J.-L. Foulley.
\newblock On testing a class of restricted hypotheses.
\newblock \emph{Journal of statistical planning and inference}, 137\penalty0
  (4):\penalty0 1343--1361, 2007.

\bibitem[Drikvandi et~al.(2013)Drikvandi, Verbeke, Khodadadi, {Partovi Nia},
  and Nia]{Dri13}
R.~Drikvandi, G.~Verbeke, A.~Khodadadi, V.~{Partovi Nia}, and V.~P. Nia.
\newblock {Testing multiple variance components in linear mixed-effects
  models}.
\newblock \emph{Biostatistics}, 14\penalty0 (1):\penalty0 144--159, 2013.

\bibitem[Drton(2009)]{Drton09}
M.~Drton.
\newblock {Likelihood ratio tests and singularities}.
\newblock \emph{Annals of Statistics}, 37\penalty0 (2):\penalty0 979--1012,
  2009.

\bibitem[Fitzmaurice et~al.(2007)Fitzmaurice, Lipsitz, and Ibrahim]{Fitz07}
G.~M. Fitzmaurice, S.~R. Lipsitz, and J.~G. Ibrahim.
\newblock {A note on permutation tests for variance components in multilevel
  generalized linear mixed models}.
\newblock \emph{Biometrics}, 63\penalty0 (3):\penalty0 942--946, 2007.

\bibitem[Geyer(1994)]{Gey94}
C.~J. Geyer.
\newblock On the asymptotics of constrained m-estimation.
\newblock \emph{The Annals of Statistics}, 22\penalty0 (4):\penalty0
  1993--2010, 1994.

\bibitem[Goymann et~al.(2016)Goymann, Safari, Muck, and Schwabl]{dryad}
W.~Goymann, I.~Safari, C.~Muck, and I.~Schwabl.
\newblock Data from: Sex roles, parental care and offspring growth in two
  contrasting coucal species, 2016.
\newblock URL \url{http://dx.doi.org/10.5061/dryad.23gt0}.

\bibitem[Greven et~al.(2008)Greven, Crainiceanu, K{\"{u}}chenhoff, and
  Peters]{Grev08}
S.~Greven, C.~M. Crainiceanu, H.~K{\"{u}}chenhoff, and A.~Peters.
\newblock {Restricted Likelihood Ratio Testing for Zero Variance Components in
  Linear Mixed Models}.
\newblock \emph{Journal of Computational and Graphical Statistics}, 17\penalty0
  (4):\penalty0 870--891, dec 2008.

\bibitem[Hiriart-Urruty and Lemarechal(1996)]{Hir96}
J.~Hiriart-Urruty and C.~Lemarechal.
\newblock \emph{Convex Analysis and Minimization Algorithms I: Fundamentals}.
\newblock Grundlehren der mathematischen Wissenschaften. Springer Berlin
  Heidelberg, 1996.

\bibitem[Hiriart-Urruty and Malick(2012)]{Hir12}
J.~B. Hiriart-Urruty and J.~Malick.
\newblock {A Fresh Variational-Analysis Look at the Positive Semidefinite
  Matrices World}.
\newblock \emph{Journal of Optimization Theory and Applications}, 153\penalty0
  (3):\penalty0 551--577, 2012.

\bibitem[Kuhn and Lavielle(2005)]{Kuh05}
E.~Kuhn and M.~Lavielle.
\newblock {Maximum likelihood estimation in nonlinear mixed effects models}.
\newblock \emph{Computational Statistics \& Data Analysis}, 49\penalty0
  (4):\penalty0 1020--1038, 2005.

\bibitem[Lavielle(2014)]{Lavielle2014}
M.~Lavielle.
\newblock \emph{Mixed effects models for the population approach: models,
  tasks, methods and tools}.
\newblock CRC Press, 2014.

\bibitem[Letort et~al.(2008)Letort, Mahe, Courn{\`e}de, De~Reffye, and
  Courtois]{letort2008}
V.~Letort, P.~Mahe, P.-H. Courn{\`e}de, P.~De~Reffye, and B.~Courtois.
\newblock Quantitative genetics and functional--structural plant growth models:
  simulation of quantitative trait loci detection for model parameters and
  application to potential yield optimization.
\newblock \emph{Annals of Botany}, 101\penalty0 (8):\penalty0 1243--1254, 2008.

\bibitem[Molenberghs and Verbeke(2007)]{Mol07}
G.~Molenberghs and G.~Verbeke.
\newblock {Likelihood Ratio, Score, and Wald Tests in a Constrained Parameter
  Space}.
\newblock \emph{The American Statistician}, 61\penalty0 (1):\penalty0 22--27,
  2007.

\bibitem[Nie(2006)]{Nie06}
L.~Nie.
\newblock Strong consistency of the maximum likelihood estimator in generalized
  linear and nonlinear mixed-effects models.
\newblock \emph{Metrika}, 63\penalty0 (2):\penalty0 123--143, 2006.

\bibitem[Nie(2007)]{Nie07}
L.~Nie.
\newblock {Convergence rate of MLE in generalized linear and nonlinear
  mixed-effects models: theory and applications}.
\newblock \emph{Journal of Statistical Planning and Inference}, 137\penalty0
  (6):\penalty0 1787--1804, 2007.

\bibitem[Pinheiro and Bates(2000)]{Pin00}
J.~Pinheiro and D.~Bates.
\newblock \emph{{Mixed-Effects Models in S and S-PLUS}}.
\newblock Springer. 2000.

\bibitem[Potthoff and Roy(1964)]{Pott64}
R.~F. Potthoff and S.~Roy.
\newblock A generalized multivariate analysis of variance model useful
  especially for growth curve problems.
\newblock \emph{Biometrika}, 51\penalty0 (3-4):\penalty0 313--326, 1964.

\bibitem[Qu et~al.(2013)Qu, Guennel, and Marshall]{Qu2013}
L.~Qu, T.~Guennel, and S.~l. Marshall.
\newblock {Linear score tests for variance components in linear mixed models
  and applications to genetic association studies}.
\newblock \emph{Biometrics}, 69\penalty0 (4):\penalty0 883--892, 2013.

\bibitem[Samuh et~al.(2012)Samuh, Grilli, Rampichini, Salmaso, and
  Lunardon]{Sam12}
M.~H. Samuh, L.~Grilli, C.~Rampichini, L.~Salmaso, and N.~Lunardon.
\newblock {The Use of Permutation Tests for Variance Components in Linear Mixed
  Models}.
\newblock \emph{Communications in Statistics - Theory and Methods}, 41\penalty0
  (16-17):\penalty0 3020--3029, 2012.

\bibitem[Self and Liang(1987)]{Self87}
S.~G. Self and K.-Y. Liang.
\newblock Asymptotic properties of maximum likelihood estimators and likelihood
  ratio tests under nonstandard conditions.
\newblock \emph{Journal of the American Statistical Association}, 82\penalty0
  (398):\penalty0 605--610, 1987.

\bibitem[Shapiro(1985)]{Sha85}
A.~Shapiro.
\newblock Asymptotic distribution of test statistics in the analysis of moment
  structures under inequality constraints.
\newblock \emph{Biometrika}, 72:\penalty0 133--144, 1985.

\bibitem[Shapiro(1988)]{Sha88}
A.~Shapiro.
\newblock {Towards a Unified Theory of Inequality Constrained Testing in
  Multivariate Analysis}.
\newblock \emph{International Statistical Review / Revue Internationale de
  Statistique}, 56\penalty0 (1):\penalty0 pp. 49--62, 1988.

\bibitem[Silvapulle(1994)]{Silv94}
M.~J. Silvapulle.
\newblock {On tests against one-sided hypotheses in some generalized linear
  models.}
\newblock \emph{Biometrics}, 50\penalty0 (3):\penalty0 853--858, 1994.

\bibitem[Silvapulle and Sen(2011)]{Sil11}
M.~J. Silvapulle and P.~K. Sen.
\newblock \emph{Constrained statistical inference: Order, inequality, and shape
  constraints}, volume 912.
\newblock John Wiley \& Sons, 2011.

\bibitem[Silvapulle and Silvapulle(1995)]{Sil95}
M.~J. Silvapulle and P.~Silvapulle.
\newblock A score test against one-sided alternatives.
\newblock \emph{Journal of the American Statistical Association}, 90\penalty0
  (429):\penalty0 342--349, 1995.

\bibitem[Sinha(2009)]{Sin09}
S.~K. Sinha.
\newblock Bootstrap tests for variance components in generalized linear mixed
  models.
\newblock \emph{Canadian Journal of Statistics}, 37\penalty0 (2):\penalty0
  219--234, 2009.

\bibitem[Stram and Lee(1994)]{stramlee94}
D.~O. Stram and J.~W. Lee.
\newblock {Variance components testing in the longitudinal mixed effects
  model.}
\newblock \emph{Biometrics}, 50\penalty0 (4):\penalty0 1171--1177, 1994.

\bibitem[Stram and Lee(1995)]{stramlee95}
D.~O. Stram and J.~W. Lee.
\newblock {Corrections to Variance components testing in the longitudinal mixed
  effects model.}
\newblock \emph{Biometrics}, 51\penalty0 (3):\penalty0 1196, 1995.

\bibitem[Tardieu(2003)]{Tar03}
F.~Tardieu.
\newblock Virtual plants: modelling as a tool for the genomics of tolerance to
  water deficit.
\newblock \emph{Trends in plant science}, 8\penalty0 (1):\penalty0 9--14, 2003.

\bibitem[Wood(2013)]{Wood13}
S.~N. Wood.
\newblock {A simple test for random effects in regression models}.
\newblock \emph{Biometrika}, 100\penalty0 (4):\penalty0 1005--1010, 2013.

\end{thebibliography}


\end{document}